\documentclass{amsart}

\usepackage{amssymb}
\usepackage{amsmath}
\usepackage{amsthm}
\usepackage[figure]{algorithm2e}
\usepackage[noend]{algpseudocode}

\usepackage{booktabs}
\usepackage{enumerate}
\usepackage{tikz}
\usetikzlibrary{automata,backgrounds,calc,positioning,decorations.pathmorphing}
\usepackage{url}
\usepackage[utf8]{inputenc}

\newtheorem{theorem}{Theorem}
\numberwithin{theorem}{section}
\newtheorem{example}[theorem]{Example}
\newtheorem{lemma}[theorem]{Lemma}

\newtheorem{proposition}[theorem]{Proposition}
\newtheorem{claim}{Claim}
\usepackage{etoolbox}
\AtEndEnvironment{proof}{\setcounter{claim}{0}}
\theoremstyle{remark}
\newtheorem*{claimproof}{Proof}

\title{The smallest grammar problem revisited}

\author[H.~Bannai]{Hideo Bannai}
\author[M.~Hirayama]{Momoko Hirayama}
\author[D.~Hucke]{Danny Hucke}
\author[S.~Inenaga]{Shunsuke Inenaga}
\author[A.~Je{\.{z}}]{Artur Je{\.{z}}}
\author[M.~Lohrey]{Markus Lohrey}
\author[C.~P.~Reh]{Carl Philipp Reh}
\email{hideo.bannai@gmail.com,
hucke@eti.uni-siegen.de,
inenaga@inf.kyushu-u.ac.jp,
aje@cs.uni.wroc.pl,
lohrey@eti.uni-siegen.de,
reh@eti.uni-siegen.de}
\thanks{A short version of this paper appeared in the Proceedings of SPIRE 2016 \cite{HuLoRe17}. \\
This work has been supported by the DFG research project
LO 748/10-1 (QUANT-KOMP)}

\date{}

\newcommand{\bin}{\mathrm{bin}}
\newcommand{\mc}{\mathcal}
\newcommand{\bb}{\mathbb}
\newcommand{\val}{\mathrm{val}}
\newcommand{\NP}{\mathsf{NP}}

\newcommand{\bisection}{{\sf BISECTION}}
\newcommand{\bi}{{\sf BI}}
\newcommand{\lzse}{{\sf LZ78}}
\newcommand{\repair}{{\sf RePair}}
\newcommand{\sequitur}{{\sf SEQUITUR}}

\newcommand{\bigO}{\mc O}

\newcommand{\rsym}[1]{#1}

\begin{document}

\begin{abstract}
In a seminal paper of Charikar et al.~({\em IEEE Transactions on Information Theory}, 51(7):2554--2576, 2005) on the smallest grammar problem, the authors derive upper and lower bounds on the approximation
ratios for several grammar-based
compressors, but in all cases there is a gap between the lower and upper bound. Here the gaps for \lzse{} and \bisection{} are closed
by showing that the approximation ratio of \lzse{} is $\Theta( (n/\log n)^{2/3})$, whereas the approximation ratio of \bisection{}
is $\Theta(\sqrt{n/\log n})$.
In addition, the lower bound for \repair{} is improved from $\Omega(\sqrt{\log n})$ to $\Omega(\log n/\log\log n)$.
%This lower bound is shown for words over binary alphabets while the old lower bound used an unbounded alphabet.
Finally, results of Arpe and Reischuk relating grammar-based compression for arbitrary
alphabets and binary alphabets are improved.

\end{abstract}

\maketitle

\section{Introduction}

\subsection{Grammar-based compression}

The idea of grammar-based compression is based on the fact that in many cases a word $w$ can be succinctly
represented by a context-free grammar that produces exactly $w$. Such a grammar is called a {\em straight-line
program} (SLP for short) for $w$. For instance, $S \to c A A B B$, $A \to aab$, $B \to CC$, $C \to cb$ is an SLP
for the word $caabaabcbcbcbcb$.
SLPs were introduced independently by various authors in different contexts~\cite{rubin76,ZiLe78,BerstelB87,Diw86}
and under different names. For instance, in \cite{BerstelB87,Diw86} the term {\em word chains} was used since SLPs
generalize addition chains from numbers to words.
Probably the best known example of a grammar-based compressor
is the classical \lzse{}-compressor of Lempel and Ziv \cite{ZiLe78}. Indeed, it is straightforward to transform
the \lzse{}-representation of a word $w$ into an SLP for $w$. Other well-known grammar-based compressors
are \bisection{} \cite{KiefferYNC00}, \sequitur{} \cite{Nevill-ManningW97}, and \repair{} \cite{DBLP:conf/dcc/LarssonM99}, just to mention a few.

A central question asked from the very beginning in the area of grammar-based compression is
how to measure the quality of an SLP,
or, more broadly, the quality of the grammar-based compressor that computes an SLP for a given input word.
One can distinguish two main approaches for such quality measures: (i) bit-based approaches, where one analyzes the
bit length of a suitable binary encoding of an SLP and (ii) size-based approaches which measure the quality of an SLP
by its size. The size of an SLP is defined as the sum of the lengths of all right-hand sides of the SLP (the SLP from
the previous paragraph has size 12). Let us briefly survey the literature on these two approaches before we explain
our main results in Section~\ref{sec-results}.

\subsection{Bit-based approaches}
It seems that the first attempt at evaluating a gram\-mar-based compressor was done for \lzse{} by Ziv and Lempel~\cite{ZiLe78},
who developed their own methodology of comparing (finite state) compressors:
In essence, given a word $w$ define $L_s(w)$ as the length of an appropriate bit encoding of the output produced by \lzse{} with window-size $s$ on input $w$
and by $L_s^*(w)$ the smallest bit-size achievable by a finite-state compressor with $s$ states on input $w$.
It was shown that $\lim_{s \to \infty} \limsup_{n \to \infty} L_s(w)/L_s^*(w) = 1$.
In other words, \lzse{} is optimal (up to lower order terms) among finite-state compressors.
(Note that the actual statement is more general, as it allows $w$ to be compressed after some initially read prefix,
i.e., we compare how $w$ in $uw$ is compressed by \lzse{} and other compressors).

Later, a systematic evaluation of grammar-based compressors was done using the information theoretic paradigm.
In \cite{KiYa00,KiefferYNC00,KiYa02,YangK00}, grammar-based compressors have been used in order to construct universal codings in the following sense:
for every finite state source and every input string $w$ of length $n$ (that is emitted with non-zero probability by the source),
the coding length of $w$ is bounded by $-\log_2 P(w) + R(n)$, where $P(w)$ is the probability that the source emits $w$
(thus $-\log_2 P(w)$ is the self-information of $w$) and $R(n)$ is a function in $o(n)$.  The function $R(n)/n$ is called the redundancy;
it converges to zero. In \cite{KiYa00,KiefferYNC00,KiYa02,YangK00}
the code for $w$ is constructed in two steps: First, an SLP is computed for $w$ using a grammar-based compressor. In a second step this SLP is encoded
by a bit string using a suitable binary encoding (see also \cite{TabeiTS13} for the problem of encoding SLPs within the information-theoretic limit).
In \cite{KiYa00} it was shown that the redundancy can be bounded by $\bigO(\log \log n/\log n)$ provided the
grammar-based compressor produces an SLP of size $\bigO(n/\log n)$ for every input string of length $n$ (this assumes an alphabet of constant size; otherwise
an additional factor $\log \sigma$ enters the bounds). The size bound  $\bigO(n/\log n)$ holds for all grammar-based compressors that produce so-called irreducible
SLPs \cite{KiYa00}, which roughly speaking means that certain redundancies in the SLP are eliminated. Moreover, every SLP can be easily made irreducible by a simple post-processing \cite{KiYa00}.
In \cite{KiYa02}, the redundancy bound from  \cite{KiYa00}  was improved to $\bigO(1/\log n)$ for so-called structured grammar-based codes.

Recently, bounds in terms of the $k$-th order empirical entropy $H_k(w)$ of the input string $w$
have been shown for grammar-based compressors~\cite{Ganczorz18,OchoaN19}. Again, these results assume a suitable
binary encoding of SLPs.
In \cite{OchoaN19} it was shown that the length of the binary encoding (using the encoding from \cite{KiYa00}) of an irreducible
SLP for a string $w$ can be bounded by $H_k(w) \cdot |w| +  \mathcal O(n k \log \sigma/\log_\sigma n)$,
where $n$ is the length of the input string and $\sigma$ is the size of the alphabet.
Note that the additional additive term $\mathcal O(n k \log \sigma/\log_\sigma n)$
 is in $o(n \log \sigma)$ under the standard assumption that
$k = o(\log_\sigma n)$.
In \cite{Ganczorz18} similar bounds are derived for more natural binary encodings of SLPs.
On the other hand, a lower bound of $H_k(w) \cdot |w| + \Omega(n k \log \sigma/\log_\sigma n)$
was recently shown for a wide class of ``natural'' grammar-based compressors~\cite{Ganczorz19}. Hence,
the mentioned upper bounds from~\cite{Ganczorz18,OchoaN19} are tight.

\subsection{Size-based approaches}

Bit-based approaches analyze the length of the binary encoding of the SLP.
For this, one has to fix a concrete binary encoding.
In contrast, the size of the SLP (the sum of the lengths of all right-hand sides)
abstracts away from the concrete binary encoding of the SLP. Analyzing this more abstract quality measure has also some advantages:
SLPs turned out to be particularly useful for the algorithmic processing of compressed data.
For many algorithmic problems on strings, efficient algorithms are known
in the setting where the input strings are represented by SLPs, see \cite{BilleLRSSW15,Loh12survey} for some examples. For the running time of these algorithms, the size of the input SLPs is the main parameter whereas
the concrete binary encoding of the SLPs is not relevant.
Another research direction where only the SLP size is relevant
arises from the recent work on string attractors, where the size of a smallest SLP for a string is compared with other string parameters that arise from dictionary
compression (number of phrases in the {\sf LZ77} parse, minimal number of phrases in a bidirectional parse, number of runs in the Burrows-Wheeler transform)
\cite{KempaP18}.

Another important aspect when comparing the bit-based approach (in particular, entropy bounds for binary encoded SLPs)
and the size-based approach was also emphasized in \cite[Section~VI]{CLLLPPSS05}: entropy bounds are often no longer useful when low-entropy
strings are considered; see also \cite{KosarajuM99} for an investigation in the context of Lempel-Ziv compression. Consider for instance
the entropy bounds in \cite{Ganczorz18,OchoaN19}. Besides the $k$-th order empirical entropy of the input string, these bounds also contain an additive
term of order $\bigO(n k \log \sigma/\log_\sigma n)$ (and by the result from \cite{Ganczorz19} this is unavoidable).
Similar remarks apply to the redundancy bound in \cite{KiYa00},
where the output bit length of the grammar-based compressor is bounded by the self-information of the input string (with respect to a $k$-th order finite state source)
plus a term of order $\bigO(n (k+\log\log_\sigma n)/\log_\sigma n)$.
For input strings with low entropy/self-information these additive terms can be much larger than the entropy/self-information. For such input strings
the existing entropy/redundancy bounds do not make useful statements about the performance
of a grammar-based compressor.

A first investigation of the SLP size was done by Berstel and Brlek~\cite{BerstelB87},
who proved that the function $g(\sigma,n) = \max \{ g(w) \mid w \in \{1,\ldots,\sigma\}^n \}$,
where $g(w)$ is the size of a smallest SLP for the word $w$,
is in $\Theta(n/\log_\sigma n)$.
Note that $g(\sigma,n)$ measures the worst case SLP-compression
over all words of length $n$ over an alphabet of size $\sigma$.
It is worth noting that addition chains~\cite{Yao76} are basically SLPs over a singleton alphabet and
that $g(1,n)$ is the size of a smallest addition chain for $n$ (up to a constant factor).

Constructing a smallest SLP for a given input word
is known as the {\em smallest grammar problem}.
%The above measure focus on grammar size for all words of some fixed length.
%Instead one can also try want to compute the smallest SLP for a word given as an input,
%this is known as the smallest grammar problem.
Storer and Szymanski~\cite{StorerS82} and Charikar et al.~\cite{CLLLPPSS05} proved that it cannot be solved in polynomial time unless
$\mathsf{P} = \mathsf{NP}$.
Moreover, Charikar et al.~\cite{CLLLPPSS05} showed that, unless $\mathsf{P} = \mathsf{NP}$,
one cannot compute in polynomial
time for a given word $w$ an SLP of size $<(8569/8568) \cdot g(w)$.
The construction in \cite{CLLLPPSS05}
uses an alphabet of unbounded size, and it was unknown whether this lower bound holds also for words over
a fixed alphabet.
In \cite{CLLLPPSS05} it is stated that the construction in~\cite{StorerS82} shows that the smallest
grammar problem for words over  a ternary alphabet cannot be solved in polynomial time unless  $\mathsf{P} = \mathsf{NP}$.
But this is not clear at all, see the recent paper \cite{CaFeGaGrSchmi16} for a detailed explanation. In the same paper
 \cite{CaFeGaGrSchmi16} it was shown that the smallest
grammar problem for an alphabet of size 24 cannot be solved in polynomial time unless $\mathsf{P} = \mathsf{NP}$ using a rather
complex construction. It is far from clear whether this construction can be adapted so that it works also for a binary alphabet.
Another idea for showing $\mathsf{NP}$-hardness of the smallest grammar problem for binary words is to
reduce the smallest grammar problem for unbounded alphabets to the smallest grammar problem for a
binary alphabet. This route was investigated in \cite{ArpeR06}, where the following result was shown for every constant $c$: If there is a polynomial time
grammar-based compressor that computes an SLP of size $c \cdot g(w)$ for a given binary input word $w$, then for every  $\varepsilon > 0$ there is a polynomial
time grammar-based compressor that computes an SLP of size $(24c+\varepsilon) \cdot g(w)$ for a given input word $w$ over an arbitrary alphabet.
The construction in \cite{ArpeR06} uses a quite technical block encoding of arbitrary alphabets into a binary alphabet.

A size-based quality measure for grammar-based compressors is the approximation ratio \cite{CLLLPPSS05}:
For a given grammar-based compressor $\mathcal{C}$ that computes from a given word $w$ an SLP $\mathcal{C}(w)$ for $w$
one defines the approximation ratio of $\mathcal{C}$ on $w$ as the quotient of the size of  $\mathcal{C}(w)$ and
the size $g(w)$ of a smallest SLP for $w$. The approximation ratio $\alpha_{\mathcal{C}}(n)$ is the maximal approximation
ratio of $\mathcal{C}$ among all words of length $n$ over any alphabet.
The approximation ratio is a useful measure for the worst-case performance of a grammar-based compressor, where
the worst-case over all strings of a certain length is considered.
This includes also low-entropy strings, for which the existing entropy/redundancy bounds are no longer useful as argued above.
In this context one should also emphasize the fact that the entropy/redundancy bounds from  \cite{KiYa00,OchoaN19} apply to all grammar-based compressors that produce irreducible SLPs.
As mentioned above, this property can be easily enforced by a simple post-processing of the SLP. This shows that the entropy/redundancy bounds
from \cite{KiYa00,OchoaN19} are not useful for a fine-grained comparison of grammar-based compressors. In contrast, analyzing the approximation ratio can lead to such a fine-grained comparison.

Charikar et al.~\cite{CLLLPPSS05} initiated a systematic investigation of the approximation ratio of various grammar-based compressors
(\lzse{}, \bisection{}, {\sf Sequential}, \repair{}, {\sf LongestMatch}, {\sf Greedy}). They proved lower and upper bounds for the approximation ratios
of theses compressor, but for none of them the lower and upper bounds match.
Moreover, Charikar et al.~present a linear time grammar-based compressor with an approximation
ratio of $\bigO(\log n)$. Other linear time grammar-based compressors which achieve the same approximation ratio can be found in
\cite{Jez13approx,Jez16,Ryt03,Sakamoto05}.  It is unknown whether there exist
grammar-based compressors that work in polynomial time and have an approximation ratio of $o(\log n)$. Getting a polynomial time grammar-based compressor
with an approximation ratio of $o(\log n/\log\log n)$ would solve a long-standing open problem on addition chains \cite{CLLLPPSS05,Yao76}.

\subsection{Results of the paper} \label{sec-results}

Our first main contribution (Section~\ref{sec-approx})
is an improved analysis of the approximation ratios of
\lzse{}, \bisection{}, and \repair{}.
These compression algorithms are among the most popular grammar-based compressors.
\lzse{} is a classical algorithm and the basis of several widely used text compressors such as {\sf LZW} (Lempel-Ziv-Welch).
\repair{}  shows in many applications the best compression results among the tested grammar-based compressors \cite{BilleGP17,GaIMaNaSaTa}
%(in particular it outperformed in practical experiments the above mentioned compressors with an approximation of $\bigO(\log n)$ \cite{})
and found applications, among others, in web graph compression \cite{ClaudeN10},
different scenarios related to word-based text compression \cite{Wan03},
searching compressed text \cite{Kida03},
suffix array compression \cite{Gonzalez07}
and (in a slightly modified form in) XML compression \cite{LohreyMM13}.
Some variants and improvements of \repair{} can be found in \cite{BilleGP17,FuruyaTNIBK19,GaIMaNaSaTa,GanczorzJ17,MasakiK16}.
\bisection{} was first studied in the context of universal lossless compression \cite{KiefferYNC00} (called {\sf MPM} there).
On bit strings of length $2^n$, \bisection{} produces in fact the ordered binary decision
diagram (OBDD) of the Boolean function represented by the bit string; see also \cite{KiFlaYa11}.
OBDDs are a widely used data structure in the area of hardware verification.

For \lzse{} and \bisection{} we close the gaps for the approximation ratio from \cite{CLLLPPSS05}.
For this we improve the corresponding lower bounds from \cite{CLLLPPSS05} and obtain the approximation ratios
$\Theta( (n/\log n)^{1/2})$ for \bisection{} and $\Theta( (n/\log n)^{2/3})$ for \lzse{}.  We prove both lower bounds using a binary alphabet.
These are the first exact (up to constant factors) approximation ratios for practical grammar-based compressors.
We also improve the lower bound for \repair{} from $\Omega\left(\sqrt{\log n}\right)$ to $\Omega\left(\log n/\log\log n\right)$ using a binary alphabet. The previous lower bound from \cite{CLLLPPSS05} used a family of words over an alphabet of unbounded size. Our new lower bound for \repair{} is still quite far away from the best known upper bound of $\bigO( (n/\log n)^{2/3})$ \cite{CLLLPPSS05}. On the other
hand, the lower bound $\Omega\left(\log n/\log\log n\right)$ is of particular interest, since it was shown in \cite{CLLLPPSS05} that
a grammar-based compressor with an approximation ratio of $o(\log n/\log\log n)$ would improve Yao's method for computing a smallest addition
chain for a set of numbers \cite{Yao76}, which is a long standing open problem. Our new lower bound excludes \repair{} as a candidate for improving Yao's method.
Let us also remark that \repair{} belongs to the class of so-called {\em global grammar-based compressors} (other examples are {\sf LongestMatch} \cite{KiYa00} and
{\sf Greedy} \cite{ApostolicoL98}). Analyzing the approximation ratio of global algorithms seems to be very difficult. We can quote here Charikar et al.~\cite{CLLLPPSS05}:
``Because they [global algorithms] are so natural and our understanding is so incomplete, global algorithms are one of the most interesting topics related to the smallest grammar problem that deserve further investigation.''
In the specific context of singleton alphabets, a detailed investigation of the approximation ratios of global grammar-based compressors was recently examined in~\cite{Hu19}.

Our second main contribution deals with the hardness of the smallest grammar problem for words over a binary alphabet.
As mentioned above, it is open whether this problem is {\sf NP}-hard. This is one of the most intriguing unsolved problems in the area of grammar-based compression.
Recall that Arpe and Reischuk \cite{ArpeR06} used a quite technical block encoding to show that if there is a polynomial time
grammar-based compressor with approximation ratio $c$ (a constant) on binary words, then there is a polynomial
time grammar-based compressor with approximation ratio $24c+\varepsilon$ for every  $\varepsilon > 0$ on arbitrary words.
Here, we present a very simple construction, which encodes the $i$-th alphabet symbol by $a^i b$, and yields the same result
as in \cite{ArpeR06} but with $24c+\varepsilon$ replaced by $6c$. In order to show {\sf NP}-hardness of the smallest grammar problem
for binary strings, one would have to reduce the factor 6 to at most $8569/8568$.
This follows from the inapproximability result for the smallest grammar problem from \cite{CLLLPPSS05}.

\section{Straight-line Programs}

Let $w=a_1\cdots a_n$ ($a_1,\dots,a_n\in\Sigma$) be a \emph{word} over an \emph{alphabet} $\Sigma$.
The length $|w|$ of $w$ is $n$ and we denote by $\varepsilon$ the word of length $0$.  Let $\Sigma^+ = \Sigma^* \setminus \{\varepsilon\}$
be the set of nonempty words.
For $w \in \Sigma^+$, we call $v\in\Sigma^+$ a \emph{factor} of $w$ if there exist $x,y\in\Sigma^*$ such that $w=xvy$.
If $x=\varepsilon$ (respectively $y=\varepsilon$) then we call $v$ a \emph{prefix} (respectively \emph{suffix}) of $w$.
A factorization of $w$ is a decomposition $w=f_1\cdots f_\ell$ into factors $f_1,\dots, f_\ell$.
For words $w_1,\dots, w_n\in\Sigma^*$, we further denote by $\prod_{i=j}^nw_i$ the word $w_jw_{j+1}\cdots w_n$ if $j\le n$ and $\varepsilon$ otherwise.

A \emph{straight-line program}, briefly SLP, is a context-free grammar that
produces a single word $w\in\Sigma^+$.
Formally, it is a tuple $\bb A = (N,\Sigma, P, S)$, where $N$ is a
finite set of nonterminals with $N\cap \Sigma = \emptyset$,
$S \in N$ is the start nonterminal, and $P$ is a finite
set of productions (or rules) of the form $A \to w$ for $A \in N$, $w \in (N \cup \Sigma)^+$ such that:
(i) For every $A \in N$, there exists exactly one production of the form $A \to w$, and
(ii) the binary relation $\{ (A, B) \in N \times N \mid (A \to w) \in P,\;B \text{ occurs in } w \}$ is acyclic.
Every nonterminal $A \in N$ produces a unique string $\val_{\bb A}(A) \in \Sigma^+$.
The string defined by $\bb A$ is $\val(\bb A) = \val_{\bb A}(S)$.
We omit the subscript $\bb A$ when it is clear from the context.
The \emph{size} of the SLP $\bb A$ is
$|\bb A| = \sum_{(A \to w) \in P} |w|$.
We denote by $g(w)$ the size of a smallest SLP producing the word $w\in\Sigma^+$.
It is easy to see that $g(w)\le |w|$ since for each word $w$ there is a trivial SLP with the only rule $S\to w$.
We will use the following lemma which summarizes known results about SLPs.

\begin{lemma} \label{lemma:folklore}
Let $\Sigma$ be a finite alphabet of size $\sigma$.
\begin{enumerate}
\item\label{nlogn} For every word $w\in\Sigma^+$ of length $n$, there exists an SLP $\bb A$ of size $\bigO\big(\frac{n}{\log_{\sigma} n}\big)$ such that $\val(\bb A)=w$.
\item\label{logn}  For an SLP $\bb A$ and a number $n>0$,
there exists an SLP $\bb B$ of size $|\bb A|+\bigO(\log n)$ such that $\val(\bb B)=\val(\bb A)^n$.
\item\label{concat} For SLPs $\bb A_1$ and $\bb A_2$ there exists an SLP $\bb B$ of size $|\bb A_1|+|\bb A_2|$ such that $\val(\bb B)=\val(\bb A_1)\val(\bb A_2)$.
\item\label{substring}  For given words $w_1,\dots,w_n\in\Sigma^*$, $u\in\Sigma^+$ and SLPs
$\bb A_1, \bb A_2$ with $\val(\bb A_1)=u$ and $\val(\bb A_2)=
w_1 x w_2 x \cdots w_{n-1} x w_n$ for a symbol $x \not\in \Sigma$,
there exists an SLP $\bb B$ of size $|\bb A_1| + |\bb A_2|$ such that $\val(\bb B)=w_1 u w_2 u \cdots w_{n-1} u w_n$.
\item\label{mk} A string $w\in\Sigma^*$ contains at most $g(w) \cdot k$ distinct factors of length $k$.
\end{enumerate}
\end{lemma}
Statement~\ref{nlogn} can be found for instance in \cite{BerstelB87}.
Statements~\ref{logn}, \ref{concat} and \ref{mk} are shown in \cite{CLLLPPSS05}.
The proof of~\ref{substring} is straightforward: Simply replace in the SLP $\bb A_2$ every occurrence of the terminal $x$ by
the start nonterminal of $\bb A_1$ and add all rules of $\bb A_1$ to $\bb A_2$.

The maximal size of a smallest SLP for all words of length $n$ over an alphabet of size $k$ is
\[
g(k,n)=\max \{ g(w) \mid w\in [1,k]^n \},
\]
where $[1,k] = \{1, \ldots, k \}$. By point~\ref{nlogn} of Lemma~\ref{lemma:folklore} we have $g(k,n) \in \bigO(n/\log_k n)$.
In fact, Berstel and Brlek proved in \cite{BerstelB87} that $g(k,n) \in \Theta(n/\log_k n)$.
As a first minor result, we show that there are words of length $2k^2+2k+1$ over an alphabet of size $k$ for which the size of a smallest SLP equals the word length. Additionally, we show that all longer words have strictly smaller SLPs. Together this yield the following proposition:
\begin{proposition} \label{prop-incompr}
Let $n_k=2k^2+2k+1$ for $k>0$. Then
(i) $g(k,n)<n$ for $n>n_k$ and (ii) $g(k,n)=n$ for $n\le n_k$.
\end{proposition}

\begin{proof}
Let $\Sigma_k=\{a_1,\ldots,a_k\}$ and
let $M_{n,\ell} \subseteq \Sigma_k^*$ be the set of all words $w$ where a factor $v$ of length $\ell$ occurs at least $n$ times without overlap.
It is easy to see that $g(w) < |w|$ if and only if $w\in M_{3,2}\cup M_{2,3}$.
Hence, we have to show that every word $w \notin M_{3,2}\cup M_{2,3}$ has length at most $2k^2+2k+1$.
Moreover, we present words $w_k \in {\Sigma_k}^*$ of length $2k^2+2k+1$ such that $w_k\notin M_{3,2}\cup M_{2,3}$.

Let $w \notin M_{3,2}\cup M_{2,3}$. Consider a factor $a_ia_j$ of length two. If $i \neq j$ then this factor does
not overlap itself, and thus $a_ia_j$ occurs at most twice in $w$. Now consider $a_i a_i$.
Then $w$ contains at most four (possibly overlapping) occurrence of $a_i a_i$, because
five occurrences of $a_ia_i$ would yield at least three non-overlapping occurrences of $a_ia_i$.
It follows that $w$ has at most $2(k^2-k)+4k$ positions where a factor of length $2$ starts, which implies
$|w| \leq 2k^2+2k+1$.

Now we create a word $w_k\notin M_{3,2}\cup M_{2,3}$ which realizes the above maximal occurrences of factors of length $2$:
\[
w_k=\left(\prod_{i=1}^{k} a_{k-i+1}^5\right)\prod_{i=1}^{k-1}\left(\prod_{j=i+2}^{k}\left(a_ja_i\right)^2\right)a_{i+1}a_ia_{i+1}
\]
For example we have $w_3= a_3^5 a_2^5 a_1^5 (a_3a_1)^2 a_2a_1a_2a_3a_2a_3$.
One can check that $|w_k|=2k^2+2k+1$ and $w_k \notin M_{3,2}\cup M_{2,3}$.
\end{proof}

\section{Approximation ratio} \label{sec-approx}

As mentioned in the introduction, there is no polynomial time algorithm that computes
a smallest SLP for a given word, unless $\mathsf{P} = \NP$ \cite{CLLLPPSS05,StorerS82}.
This result motivates approximation algorithms which are called \emph{grammar-based compressors}.
A grammar-based compressor $\mc C$ computes for a word $w$ an SLP $\mc C(w)$ such that $\val(\mc C(w))=w$.
%We denote by $\mc C(w)$ the SLP produced by $\mc C$ on input $w$.
The \emph{approximation ratio} $\alpha_{\mc C}(w)$ of $\mc C$ for an input $w$ is defined as $|\mc C(w)|/g(w)$.
The worst-case approximation ratio $\alpha_{\mc C}(k,n)$ of $\mc C$ is the maximal approximation ratio over
all words of length $n$ over an alphabet of size $k$:
\[\alpha_{\mc C}(k,n)=\max \{ \alpha_{\mc C}(w) \mid w \in [1,k]^n \} = \max\{ |\mc C(w)|/g(w) \mid w \in [1,k]^n \} \]
%It is easy to see that $\alpha_{\mc C}(k,n)\ge\alpha_{\mc C}(w)$ for each word $w$ of length $n$ over an alphabet of size $k$.
%We will use this property to establish lower bounds for the worst-case approximation ratio by analysing particular families of words.
In this definition, $k$ might depend on $n$. Of course we must have $k \leq n$ and we write $\alpha_{\mc C}(n)$ instead of $\alpha_{\mc C}(n,n)$.
This corresponds to the case where there is no restriction on the alphabet at all and it is the
definition of the worst-case approximation ratio in \cite{CLLLPPSS05}.
The grammar-based compressors studied in our work are \bisection{}~\cite{KiefferYNC00}, \lzse{}~\cite{ZiLe78} and \repair{} \cite{DBLP:conf/dcc/LarssonM99}.
We will abbreviate the approximation ratio of \bisection{} by $\alpha_\bi{}$.
%In the following we will improve the known results on the worst-case approximation ratio of \bisection{} and \lzse{}.
The families of words which we will use to improve the lower bounds of $\alpha_{\bi{}}(n)$ and $\alpha_{\lzse{}}(n)$ are inspired by the
constructions in \cite{CLLLPPSS05}.

\subsection{\bisection{}}

The \bisection{} algorithm~\cite{KiefferYNC00} first splits an input word $w$ with $|w| \geq 2$ as
$w = w_1 w_2$ such that
$|w_1|=2^j$ for the unique number $j \geq 0$ with $2^j<|w|\le 2^{j+1}$.
This process is recursively repeated with $w_1$ and $w_2$ until we obtain words of length $1$.
During the process, we introduce a nonterminal for each distinct factor of length at least two and create a rule with two symbols on the right-hand
 side corresponding to the split.
 Note that if $w = u_1 u_2 \cdots u_k$ with $|u_i| = 2^n$ for all $i,1 \leq i \leq k$, then the SLP produced by
 \bisection{} contains a nonterminal for each distinct word $u_i$ ($1\le i\le k$).

\begin{example}\label{example-bisection}
\bisection{} constructs an SLP for $w=ababbbaabbaaab$ as follows:
\begin{itemize}
\item $w = w_1w_2$ with $w_1=ababbbaa$, $w_2=bbaaab$\newline Introduced rule: $S\to W_1W_2$
\item $w_1 = x_1 x_2$ with $x_1=abab$, $x_2=bbaa$, and $w_2 = x_2 x_3$ with
$x_3=ab$\newline Introduced rules: $W_1\to X_1X_2,\;W_2\to X_2X_3$, $X_3 \to ab$
\item $x_1 = x_3x_3$, $x_2 = y_1 y_2$ with $y_1=bb$ and $y_2=aa$ \newline Introduced rules:
$X_1\to X_3X_3,\; X_2\to Y_1Y_2,\;Y_1\to bb,\;Y_2\to aa$
\end{itemize}
\end{example}
\bisection{} performs asymptotically optimal on unary words $a^n$ since it produces an SLP of
size $\bigO(\log n)$.
Therefore $\alpha_\bi{}(1,n)\in\Theta(1)$.
The following bounds on the approximation ratio for alphabets of size at least two are proven in \cite[Thm.~5 and 6]{CLLLPPSS05}:
\begin{eqnarray}
\alpha_\bi{}(2,n) & \in & \Omega(\sqrt{n}/\log n) \label{bisec-lower} \\
\alpha_\bi{}(n) & \in & \bigO(\sqrt{n/\log n}) \label{bisec-upper}
\end{eqnarray}
We improve the lower bound \eqref{bisec-lower} so that it matches the upper bound
 \eqref{bisec-upper}:

\begin{theorem} \label{bisection-lower-bound}
For every $k,2 \leq k \leq n$ we have
$\alpha_\bi{}(k,n) \in \Theta(\sqrt{n/\log n})$.
\end{theorem}

\begin{proof}
The upper bound \eqref{bisec-upper} implies that
$\alpha_\bi{}(k,n) \in \bigO(\sqrt{n/\log n})$ for all $k,2 \leq k \leq n$. So it suffices to show
$\alpha_\bi{}(2,n) \in \Omega(\sqrt{n/\log n})$.
We first show that $\alpha_\bi{}(3,n) \in \Omega(\sqrt{n/\log n})$.
In a second step, we encode a ternary alphabet into a binary alphabet while
preserving the approximation ratio.

For every $k\geq 2$ let $\bin_k:\{0,1,\dots,k-1\}\to\{0,1\}^{\lceil \log_2 k\rceil}$
be the function where $\bin_k(j)$ ($0\le j\le k-1$) is the binary representation of $j$ padded with leading zeros (e.g. $\bin_{9}(3)=0011$).
We further define for every $k \geq 2$ the word
\[u_k=\left(\prod_{j=0}^{k-2}\bin_k(j)a^{m_k}\right)\bin_k(k-1),\]
where $m_k=2^{k-\lceil \log_2 k\rceil}-\lceil \log_2 k\rceil$.
For instance $k=4$ leads to $m_k = 2$ and $u_4=00aa01aa10aa11$.
We analyze the approximation ratio $\alpha_\bi{}(s_k)$ for the word
\[s_k=\left(u_ka^{m_k+1}\right)^{m_k}u_k.\]
\begin{claim} \label{claim-bisec1}
  The SLP produced by \bisection{} on input $s_{k}$ has size $\Omega(2^k)$.
\end{claim}
\begin{claimproof}
If $s_k$ is split into non-overlapping factors of length $m_k+\lceil \log_2 k\rceil=2^{k-\lceil \log_2 k\rceil}$, then the resulting set $F_k$ of factors is
\[F_k=\{a^i\bin_k(j)a^{m_k-i}\mid 0\le j\le k-1,\;0\le i\le m_k\}.\]
For example $s_4$ consecutively consists of the factors $00aa$, $01aa$, $10aa$, $11aa$, $a00a$, $a01a$, $a10a$, $a11a$, $aa00$, $aa01$, $aa10$ and $aa11$.
The size of $F_k$ is $(m_k+1)\cdot k\in\Theta(2^k)$, because all factors are pairwise different and $m_k\in\Theta(2^k/k)$. It follows that the SLP produced by \bisection{} on input $s_k$ has size $\Omega(2^k)$, because the length of each factor in $F_k$ is a power of two and thus \bisection{} creates a nonterminal for each distinct factor in $F_k$.
\hspace*{\fill}  \mbox{({\em end proof of Claim~\ref{claim-bisec1}})}
\end{claimproof}
\begin{claim} \label{claim-bisec2}
  A smallest SLP producing $s_{k}$ has size $\bigO(k)$.
\end{claim}
\begin{claimproof}
There is an SLP of size $\bigO(\log m_k)=\bigO(k)$ for the word $a^{m_k}$ by Lemma~\ref{lemma:folklore} (point~\ref{logn}).
This yields an SLP for $u_k$ of size $\bigO(k)+g(u_k')$ by Lemma~\ref{lemma:folklore} (point~\ref{substring}), where $u_k'=(\prod_{i=0}^{k-2}\bin_k(i)x)\bin_k(k-1)$ is obtained from $u_k$ by replacing all occurrences of $a^{m_k}$ by a fresh symbol $x$.
The word $u_k'$ has length $\Theta(k\log k)$. Applying point~\ref{nlogn} of Lemma~\ref{lemma:folklore} (note that $u'_k$ is a word over a ternary alphabet) it follows that
\[g(u_k')\in \bigO\left(\frac{k\log k}{\log(k\log k)}\right)=\bigO\left(\frac{k\log k}{\log k+\log\log k}\right)=\bigO(k).\]
Hence $g(u_k)\in \bigO(k)$. Finally, the SLP of size $\bigO(k)$ for $u_k$ yields an SLP of size $\bigO(k)$ for $s_k$ again using Lemma~\ref{lemma:folklore} (points~\ref{logn} and~\ref{concat}).
\hspace*{\fill}  \mbox{({\em end proof of Claim~\ref{claim-bisec2}})}
\end{claimproof}
In conclusion: We showed that a smallest SLP for $s_k$ has size $\bigO(k)$, while
\bisection{} produces an SLP of size $\Omega(2^k)$.% on input $s_k$.
This implies $\alpha_\bi{}(s_k) \in \Omega(2^k/k)$.
Let $n = |s_k|$. Since $s_k$ is the concatenation of $\Theta(2^k)$ factors of length $\Theta(2^k/k)$, we have
$n \in \Theta(2^{2k}/k)$ and thus $\sqrt{n}\in\Theta(2^{k}/\sqrt{k})$. This yields $\alpha_\bi{}(s_k)\in\Omega(\sqrt{n/k})$.
Together with $k\in\Theta(\log n)$ we obtain $\alpha_\bi{}(3,n) \in \Omega(\sqrt{n/\log n})$.

Let us now encode words over $\{0,1,a\}$ into words over $\{0,1\}$.
Consider the homomorphism $f:\{0,1,a\}^*\to\{0,1\}^*$ with $f(0)=00$, $f(1)=01$ and $f(a)=10$.
Then we can prove the same approximation ratio of \bisection{} for the input $f(s_k)\in\{0,1\}^*$ that we proved for $s_k$ above:
The size of a smallest SLP for $f(s_k)$ is at most twice as large as the size of a smallest SLP for $s_k$, because an SLP for $s_k$
can be transformed into an SLP for $f(s_k)$ by replacing every occurrence of a symbol $x \in \{0,1,a\}$ by $f(x)$.
Moreover, if we split $f(s_k)$ into non-overlapping factors of twice the length as we considered for $s_k$, then
we obtain the factors from $f(F_k)$, whose length is again a power of two. Since $f$ is injective, we have
$|f(F_k)| = |F_k| \in \Theta(2^k)$.
\end{proof}

%%%%%%%%%%%%%%%%%%%%%%%%%%%%%%%%%%%%%%%%%%%%%%%%%%%%%%%%%%%%%%%%%%%%%%%%%%%%%%%%%%%%%%%%%%%%%%%%%%%%%%%%%%%%
\subsection{\lzse{}}

The \lzse{} algorithm on input $w\in\Sigma^+$ implicitly creates a list of words $f_{1},\dots,f_{\ell}$ (which we call the \emph{\lzse{}-factorization}) with $w=f_{1}\cdots f_{\ell}$ such that the following properties hold, where we set $f_0 = \varepsilon$:
\begin{itemize}
\item $f_{i}\neq f_{j}$ for all $i,j, 0\le i,j\le \ell-1$, with $i\neq j$.
\item For all $i,1\le i \le \ell-1$, there exist $j,0 \leq j<i$ and $a \in \Sigma$ such that $f_{i}=f_{j}a$.
\item $f_{\ell}=f_{i}$ for some $0 \leq i\leq \ell-1$.
\end{itemize}
Note that the \lzse{}-factorization is unique for each word $w$.
To compute it, the $\lzse{}$ algorithm needs $\ell$ steps performed by a single left-to-right pass.
In the $k^{\text{th}}$ step ($1\le k\le \ell-1$) it chooses the factor $f_k$ as the shortest prefix of the unprocessed suffix $f_k\cdots f_\ell$ such that $f_k\neq f_i$ for all $i<k$. If there is no such prefix, then the end of $w$ is reached and the algorithm sets $f_\ell$ to the (possibly empty) unprocessed suffix of $w$.

The factorization $f_1,\dots, f_\ell$ yields an SLP for $w$ of size at most $3\ell$ as described in the following example:
\begin{example}
The \lzse{}-factorization of $w=aabaaababababaa$ is
$a$, $ab$, $aa$, $aba$, $b$, $abab$, $aa$ and leads to an SLP with the following rules:
\begin{itemize}
\item $S\to F_1F_2F_3F_4F_5F_6F_3$
\item $F_1\to a,\;F_2\to F_1b,\;F_3\to F_1a,\;F_4\to F_2a,\;F_5\to b,\;F_6\to F_4b$
\end{itemize}
We have a nonterminal $F_i$ for each factor $f_i$ ($1\le i\le 6$) such that $\val_{\bb A}(F_i)=f_i$.
The last factor $aa$ is represented in the start rule by the nonterminal $F_3$.
\end{example}
The \lzse{}-factorization of $a^n$ ($n>0$) is $a^1,a^2,\dots,a^{m}, a^k$, where $k\in\{0,\dots,m\}$ such that $n=k+\sum_{i=1}^m i$.
Note that $m\in\Theta(\sqrt{n})$ and thus $\alpha_\lzse{}(1,n)\in\Theta(\sqrt{n}/\log n)$.
The following bounds for the worst-case approximation ratio of \lzse{} were shown in \cite[Thm.~3 and 4]{CLLLPPSS05}:
\begin{eqnarray}
\alpha_\lzse{}(2,n)  & \in & \Omega(n^{2/3}/\log n)  \label{lz78-lower} \\
\alpha_\lzse{}(n)  & \in & \bigO((n/\log n)^{2/3}) \label{lz78-upper}
\end{eqnarray}
We will improve the lower bound  so that it matches the upper bound in~\eqref{lz78-upper}.

\begin{theorem}\label{lz-lower-bound-bounded-alphabet}
For every $k,2 \leq k \leq n$ we have $\alpha_\lzse{}(k,n)\in\Theta((n/\log n)^{2/3})$.
\end{theorem}

\begin{proof}
Due to \eqref{lz78-upper} it suffices to show $\alpha_\lzse{}(2,n)\in\Omega((n/\log n)^{2/3})$.
For  $k \geq 2, m \geq 1$, let
%$u_{m,k}=\left(a^kb^mc\right)^{k(m+2)-1}$ and
$u_{m,k} = ((a^k b^{(2m+1)} a)^m (a^k b^{(m+1)})^2)^k a^k$ and
$v_{m,k}=\left(\prod_{i=1}^m b^ia^k\right)^{k^2}$.
We now analyze the approximation ratio of \lzse{} on the words
%\[s_{m,k}=a^{k(k+1)/2}\;b^{m(m+1)/2}\;u_{m,k}\;v_{m,k}.\]
\[s_{m,k} = a^{\frac{k(k+1)}{2}} b^{m(2m+1)}\;u_{m,k}\;v_{m,k}.\]

For example we have
%$u_{2,4}=(a^4b^2c)^{15}$,
$u_{2,4}=((a^4b^5a)^2(a^4b^3)^2)^4a^4$,
$v_{2,4}=(ba^4b^2a^4)^{16}$ and
$s_{2,4}=a^{10}\;b^{10}\; u_{2,4}\; v_{2,4}$.

\begin{claim}\label{claim:lzse-size}
  The SLP produced by $\lzse{}$ on input $s_{m,k}$ has size $\Theta(k^2 m)$.
\end{claim}
\begin{claimproof}
We consider the \lzse{}-factorization $f_1,\dots,f_\ell$ of $s_{m,k}$. Example~\ref{example-lz78} gives a complete example.
The prefix $a^{k(k+1)/2}$ produces the factors $f_i=a^i$ for every $i,1\le i \le k$
and the substring $b^{m(2m+1)}$ produces the factors $f_{k+i}=b^i$ for every $i,1\le i \le 2m$.

%%%%%%%%%%%%%%%%%%%%%%%%%%%%%%%%%%%%%%%%%%%%%%%%%%%%%%%%%%%%%%%%%%%%%%%%%%%%%%%%%%%%%%%%%%%
% u_{m,k}
%%%%%%%%%%%%%%%%%%%%%%%%%%%%%%%%%%%%%%%%%%%%%%%%%%%%%%%%%%%%%%%%%%%%%%%%%%%%%%%%%%%%%%%%%%%
We  next show that the substring $u_{m,k}$ then produces all factors from
%(among other factors) all factors $a^ib^j$, where $1\le i\le k$, $1\le j\le m$.
\begin{eqnarray}
  \bigcup_{j=1}^k \left(\{ a^{k-j+1} b^{i+1}, b^{2m-i}a^j \mid 0 \le i \le m-1\}\cup\{a^{k-j+1}b^{m+1},a^k b^{m+1}a^j \}\right).\label{eqn:ukm_factors}
\end{eqnarray}
Let
\begin{eqnarray*}
  u_{m,k,j} &=& a^{k-j+1}b^{2m+1}a(a^kb^{2m+1}a)^{m-1} (a^kb^{m+1})^2 a^j\\
  &=& (a^{k-j+1}b^{2m+1}a^j)^m a^{k-j+1}b^{m+1} a^{k} b^{m+1} a^j.
\end{eqnarray*}
Then, $u_{m,k} = u_{m,k,1} \cdots u_{m,k,k}$.
We show that each $u_{m,k,j}$ produces the factors from
\begin{eqnarray}
  \{ a^{k-j+1} b^{i+1}, b^{2m-i}a^j \mid 0 \le i \le m-1\}\cup\{a^{k-j+1}b^{m+1},a^k b^{m+1}a^j\}, \label{eqn:ukmj_factors}
\end{eqnarray}
for each $j,1\le j \le k$,
thus obtaining \eqref{eqn:ukm_factors}.

Consider
\begin{eqnarray}
  u_{m,k,1} &=& (a^kb^{2m+1}a)^ma^kb^{m+1}a^kb^{m+1}a.
\end{eqnarray}
From the factorization of the prefix $a^{\frac{k(k+1)}{2}} b^{m(2m+1)}$ of $s_{m,k}$,
the first $a^k b^{2m+1}a$ is factorized into $a^kb$ and $b^{2m}a$.
Next, for each of the following $a^k b^{2m+1}a$, we can see that the new factors are
$a^kb^2$ and $b^{2m-1}a$, $a^kb^3$ and $b^{2m-2}a, \ldots, a^k b^{m}$ and $b^{m+1}a$.
Finally, the remaining $a^kb^{m+1}a^kb^{m+1}a$ is factorized to $a^k b^{m+1}$ and  $a^k b^{m+1} a$.
Therefore, (\ref{eqn:ukmj_factors}) gives the factors of $u_{m,k,j}$ for $j = 1$.

Next, suppose that $u_{m,k,j'}$ produces the factors shown in (\ref{eqn:ukmj_factors}) for all
$1 < j' < j$, and consider $u_{m,k,j}$.
By the induction hypothesis, we see that $a^{k-j+1}b$ and $b^{2m}a^j$ are the first two factors.
Similarly, we see that each of the following $a^{k-j+1}b^{2m+1}a^j$ is factorized to
$a^{k-j+1}b^2$ and $b^{2m-1}a^j$,
$a^{k-j+1}b^3$ and $b^{2m-2}a^j, \ldots, a^{k-j+1}b^m$ and $b^{m+1}a^j$.
Finally, the remaining suffix $a^{k-j+1}b^{m+1}a^kb^{m+1}a^j$ is factorized to
$a^{k-j+1}b^{m+1}$ and $a^kb^{m+1}a^j$.
It follows that the factorization of $u_{m,k}$ yields the factors shown in
(\ref{eqn:ukm_factors}).

%%%%%%%%%%%%%%%%%%%%%%%%%%%%%%%%%%%%%%%%%%%%%%%%%%%%%%%%%%%%%%%%%%%%%%%%%%%%%%%%%%%%%%%%%%%
% v_{m,k}
%%%%%%%%%%%%%%%%%%%%%%%%%%%%%%%%%%%%%%%%%%%%%%%%%%%%%%%%%%%%%%%%%%%%%%%%%%%%%%%%%%%%%%%%%%%
Next, we will show that the remaining suffix $v_{m,k}$ of $s_{m,k}$ produces the set of factors
\[
\left\{a^ib^pa^j\mid 0\le i \le k-1,\;1\le j\le k,\;1\le p\le m\right\}.
\]
%Let $v_{m,k} = v_{m,k,1} \cdots v_{m,k,k}$ where
%\[
%  v_{m,k,j} =
%\begin{cases}
%  \phantom{a^{k-j}}(\prod_{i=1}^m b^i a^k)b a^j & j=1\\
%  a^{k-j}(\prod_{i=2}^m b^i a^k)b a^j & 1 < j < k\\
%  a^{k-j}(\prod_{i=2}^m b^i a^k) & j = k.
%\end{cases}
%\]
Observe that from the factors produced so far, only the factors $a^{j} b^{i}$ for
$0 \leq j \leq k, 0 \leq i \leq m$
can be used for the factorization of $v_{m,k}$. The reason for this is that all other
factors contain an occurrence of $b^{m+1}$, which does not occur in $v_{m,k}$.

Let $x=k+2m+k(2m+2)$ and note that this is the number of factors that we have produced so far.
The factorization of $v_{m,k}$ in $s_{m,k}$ slightly differs when $m$ is even, resp.,~odd.
We now assume that $m$ is even and explain the difference to the other case afterwards.
The first factor of $v_{m,k}$ in $s_{m,k}$ is $f_{x+1}=ba$.
We already have produced the factors $a^{k-1}b^i$ for every $i,1\le i\le m$, and hence
$f_{x+i}=a^{k-1}b^ia$ for every $i,2\le i\le m$ and $f_{x+m+1}=a^{k-1}ba$.
%For better readability, we will also use $b_i$ with $i>m$ to represent the symbol $b_{i \bmod m}$, where %$b_0=b_m$.
%For example we can use $b_{m+1}$ instead of $b_{1}$ and thus $f_{k+i}=a^kb_i$ for $2\le i\le m+1$.
The next $m$ factors are $f_{x+m+i}=a^{k-1}b^ia^2$ if $i$ is even, $f_{x+m+i}=a^{k-2}b^ia$ if $i$ is odd ($2\le i\le m$) and $f_{x+2m+1}=a^{k-2}ba$.
This pattern  continues: The next $m$ factors
are $f_{x+2m+i}=a^{k-1}b^ia^3$ if $i$ is even, $f_{x+2m+i}=a^{k-3}b^ia$ if $i$ is odd ($2\le i\le m$) and $f_{x+3m+1}=a^{k-3}ba$
and so on. %(the first $a$-block decreases in the odd case, the second $a$-block increases in the even case)
Hence, we get the following sets of factors for $(\prod_{i=1}^m b^ia^k)^{k}$:
%(which is the \lzse{}-factorization of $a^{k(k-1)/2}{u_k}^{k+1}$):
\begin{enumerate}[(i)]
%\item $\{a^i\mid 1\le i\le k\}$ for $f_1,\dots,f_k$
\item $\{a^{k-i}b^pa\mid 1\le i\le k,\;1\le p\le m,\;p \text{ is odd}\}$ for $f_{x+1},f_{x+3}\dots,f_{x+km-1}$
\item $\{a^{k-1}b^pa^j\mid 1\le j\le k,\;1\le p\le m,\;p \text{ is even}\}$ for $f_{x+2},f_{x+4},\dots, f_{x+km}$
\end{enumerate}
The remaining word then starts with the factor $f_{y+1}=ba^2$, where $y = x+km$.
Now the former pattern can be adapted to the next $k$ repetitions of $\prod_{i=1}^m b^ia^k$ which gives us %the following factors:
\begin{enumerate}[(i)]
%\item $\{a^i\mid 1\le i\le k\}$ for $f_1,\dots,f_k$
\item\label{odd} $\{a^{k-i}b^pa^2\mid 1\le i\le k,\;1\le p\le m,\;p \text{ is odd}\}$ for $f_{y+1},f_{y+3}\dots,f_{y+km-1}$
\item\label{even} $\{a^{k-2}b^pa^j\mid 1\le j\le k,\;1\le p\le m,\;p \text{ is even}\}$ for $f_{y+2},f_{y+4},\dots, f_{y+km}$
\end{enumerate}
The iteration of this process then reveals the whole pattern and thus yields the claimed factorization of $v_{m,k}$ in $s_{m,k}$ into factors $a^ib^pa^j$ for every $i,0 \le i \le k-1$, $j,1\le j\le k$ and $p,1\le p\le m$.
%\begin{enumerate}[(i)]
%\item $\{a^i\mid 1\le i\le k\}$ for $f_1,\dots,f_k$
%\item\label{factors-odd} $\{a^{k-j}b_ia^p\mid 0\le p \le r,\;0\le j\le k,\;1\le i\le m,\;i \text{ is odd}\}$
%for $f_{k+1},f_{k+3}\dots,f_{k+(r+1)\cdot(k+1)\cdot m-1}$
%\item\label{factors-even} $\{a^{k-p}b_ia^j\mid 0\le p\le r,\;0\le j\le k,\;1\le i\le m,\;i \text{ is even}\}$
%for $f_{k+2},f_{k+4},\dots, f_{k+(r+1)\cdot(k+1)\cdot m}$
%\end{enumerate}
If $m$ is odd then the patterns in~(\ref{odd}) and~(\ref{even}) switch after each occurrence of $\prod_{i=1}^m b^ia^k$, which does not affect the result but makes the pattern slightly more complicated. But the case that $m$ is even suffices in order to derive the lower bound from the theorem.

We conclude that there are exactly $k+2m+k(2m+2)+k^2m$ factors (ignoring $f_\ell=\varepsilon$) and hence the SLP produced by $\lzse{}$ on input $s_{m,k}$ has size $\Theta(k^2m)$.
\hspace*{\fill}  \mbox{({\em end proof of Claim~\ref{claim:lzse-size}})}
\end{claimproof}
\begin{claim}\label{claim:lzse-slp}
A smallest SLP producing $s_{m,k}$ has size $\bigO(\log k+m)$.
\end{claim}
\begin{claimproof}
We will combine the points stated in Lemma~\ref{lemma:folklore} to prove this claim.
Points~\ref{logn} and~\ref{concat} yield an SLP of size $\bigO(\log k+\log m)$ for the prefix $a^{k(k+1)/2}\;b^{m(2m+1)}\;u_{m,k}$ of $s_{m,k}$.
To bound the size of an SLP for $v_{m,k}$ note at first that there is an SLP of size $\bigO(\log k)$ producing $a^k$ by point~\ref{logn} of Lemma~\ref{lemma:folklore}.
Applying point~\ref{substring} and again point~\ref{logn}, it follows that there is an SLP of size $\bigO(\log k)+g(v_{m,k}')$ producing $v_{m,k}$, where $v_{m,k}'=\prod_{i=1}^m b^ix$ for some fresh letter $x$.
To get a small SLP for $v_{m,k}'$, we can introduce $m$ nonterminals $B_1,\dots,B_m$ producing $b^1,\dots,b^m$ by adding rules
$B_1\to b$ and $B_{i+1}\to B_ib$ ($1\le i\le m-1$).
This is enough to get an SLP of size $\bigO(m)$ for $v_{m,k}'$ and therefore an SLP of size $\bigO(\log k+m)$ for $v_{m,k}$.
Together with our first observation and point~\ref{concat} of Lemma~\ref{lemma:folklore} this yields an SLP of size $\bigO(\log k+m)$ for $s_{m,k}$.
\hspace*{\fill}  \mbox{({\em end proof of Claim~\ref{claim:lzse-slp}})}
\end{claimproof}
Claims \ref{claim:lzse-size} and \ref{claim:lzse-slp} imply $\alpha_\lzse{}(s_{m,k})\in \Omega(k^2m /(\log k+m))$. Let us now fix $m=\lceil\log k\rceil$.
We get $\alpha_\lzse{}(s_{m,k})\in\Omega(k^2)$. Moreover, for the length $n = |s_{m,k}|$ of $s_{m,k}$ we have
$n\in\Theta(k^3m+k^2m^2) = \Theta(k^3 \log k)$. We get
 $\alpha_\lzse{}(s_{m,k})\in\Omega((n/\log k)^{2/3})$ which together with $\log n\in\Theta(\log k)$ finishes the proof.
\end{proof}

\begin{example}  \label{example-lz78}
Here is the complete \lzse{} factorization of
\[ s_{2,4}=a^{10}b^{10}\underbrace{((a^4b^5a)^2(a^4b^3)^2)^4a^4}_{u_{2,4}}\underbrace{(ba^4b^2a^4)^{16}}_{v_{2,4}} . \]
Factors of $a^{10}$: $\;a,\;a^2,\;a^3,\;a^4$

\smallskip
\noindent
Factors of $b^{10}$: $\;b,\;b^2,\;b^3,\;b^4$

\smallskip
\noindent
Factors of $u_{2,4}$:
\[
\begin{array}{llllll}
a^4b & b^4a & a^4b^2 & b^3a & a^4b^3 & a^4b^3a \\
a^3b & b^4a^2 & a^3b^2 & b^3a^2 & a^3b^3 & a^4b^3a^2 \\
a^2b & b^4a^3 & a^2b^2 & b^3a^3 & a^2b^3 & a^4b^3a^3 \\
ab & b^4a^4 & ab^2 & b^3a^4 & ab^3 & a^4b^3a^4
\end{array}
\]
Factors of $v_{2,4}$:
\[
\begin{array}{ll}
ba & a^3b^2a \\
a^3ba & a^3b^2a^2 \\
a^2ba & a^3b^2a^3 \\
aba & a^3b^2a^4 \\
ba^2 & a^2b^2a \\
a^3ba^2 & a^2b^2a^2 \\
a^2ba^2 & a^2b^2a^3 \\
aba^2 & a^2b^2a^4 \\
ba^3 & ab^2a \\
a^3ba^3 & ab^2a^2 \\
a^2ba^3 & ab^2a^3 \\
aba^3 & ab^2a^4 \\
ba^4 & b^2a \\
a^3ba^4 & b^2a^2 \\
a^2ba^4 & b^2a^3 \\
aba^4 & b^2a^4
\end{array}
\]
\end{example}

%It remains open whether also $\alpha_\lzse{}(2,n)\in\Theta((n/\log n)^{2/3})$ holds. In contrast to
%\bisection{} it is not clear how to encode a ternary alphabet into a binary alphabet while preserving
%the approximation ratio for \lzse{}.

%%%%%%%%%%%%%%%%%%%%%%%%%%%%%%%%%%%%%%%%%%%%%%%%%%%%%%%%%%%%%%%%%%%%%%%%%%%%%%%%%%%%%%%%%%%%%%%%%%%%%%%%%%%%

\subsection{\repair{}}
For a given SLP $\bb A = (N,\Sigma, P, S)$, a word $\gamma \in (N \cup \Sigma)^+$ is called a \emph{maximal string} of $\bb A$ if
\begin{itemize}
\item $|\gamma|\ge 2$,
\item $\gamma$ appears at least twice without overlap in the right-hand sides of $\bb A$,
\item and no strictly longer word appears at least as many times on the right-hand sides of $\bb A$ without overlap.
\end{itemize}
A \emph{global grammar-based compressor} starts on input $w$ with the trivial SLP $\bb A=(\{S\},\Sigma, \{S\to w\}, S)$.
In each round, the algorithm selects a maximal string $\gamma$ of $\bb A$
and updates $\bb A$ by replacing a largest set of pairwise non-overlapping occurrences of $\gamma$ in $\bb A$ by a fresh nonterminal $X$.
Additionally, the algorithm introduces the rule $X\to \gamma$.
The algorithm stops when no maximal string occurs.
The global grammar-based compressor \repair{}~\cite{DBLP:conf/dcc/LarssonM99} selects in each round a most frequent maximal string.
Note that the replacement is not unique, e.g. the word $a^5$ with the maximal string $\gamma=aa$ yields SLPs with rules $S\to XXa, X\to aa$ or $S\to XaX, X\to aa$ or $S\to aXX, X\to aa$.
We assume the first variant in this paper, i.e. maximal strings are replaced from left to right.

The above description of \repair{} is taken from~\cite{CLLLPPSS05}. In most papers on \repair{} the algorithm works slightly different: It replaces
in each step a digram (a string of length two) with the maximal number of pairwise non-overlapping occurrences in the right-hand sides.
For example, for the string $w = abcabc$ this produces the SLP $S \to BB$, $B \to Ac$, $A \to ab$, whereas the
\repair{}-variant from \cite{CLLLPPSS05} produces the smaller SLP $S \to AA$, $A \to abc$.

The following lower and upper bounds on the approximation ratio of \repair{} were shown in~\cite{CLLLPPSS05}:
\begin{eqnarray}
\alpha_\mathsf{RePair}(n)\in\Omega\left(\sqrt{\log n}\right) \label{repair-lower}\\
\alpha_\mathsf{RePair}(2,n)\in \bigO\left((n/\log n)^{2/3}\right) \notag
\end{eqnarray}
The proof of the lower bound \eqref{repair-lower} assumes an alphabet of unbounded size.
To be more accurate, the authors construct  for every $k$ a word $w_k$ of length $\Theta(\sqrt{k} 2^k)$ over an alphabet of size
$\Theta(k)$ such that $g(w) \in \bigO(k)$ and \repair{} produces a grammar of size $\Omega(k^{3/2})$ for $w_k$.
%over an alphabet of size $\Theta(\log |w_k|)$
%with $\alpha_{\mathsf{RePair}}(w_k)\in\Omega(\sqrt{\log |w_k|})$.
We will improve this lower bound using only a binary alphabet. To do so, we first need to know how \repair{} compresses unary words.
\begin{example}[unary inputs]
\label{unary}
\repair{} produces on input $a^{27}$ the SLP with rules $X_1\to aa$, $X_2\to X_1X_1$, $X_3\to X_2X_2$ and $S\to X_3X_3X_3X_1a$, where $S$ is the start nonterminal. For the input $a^{22}$ only the start rule $S\to X_3X_3X_2X_1$ is different.
\end{example}

In general, \repair{} creates on unary input $a^m$ ($m\ge 4$) the rules $X_1\to aa$, $X_i\to X_{i-1}X_{i-1}$ for $2\le i\le \lfloor \log m\rfloor-1$ and a start rule, which is strongly related to the binary representation of $m$ since each nonterminal $X_i$ produces the word $a^{2^i}$. To be more accurate, let
$b_{\lfloor \log m\rfloor} b_{\lfloor \log m\rfloor-1}\cdots b_1b_0$
be the binary representation of $m$ and define the mappings $f_i$ ($i \geq 0$) by:
\begin{itemize}
\item $f_0:\{0,1\}\to\{a,\varepsilon\}$ with $f_0(1)=a$ and $f_0(0)=\varepsilon$,\label{f0}
\item $f_i:\{0,1\}\to \{X_i,\varepsilon\}$ with $f_i(1)=X_i$ and $f_i(0)=\varepsilon$ for $i\ge 1$. %$1\le i\le \lfloor\log m\rfloor-1$.
\end{itemize}
Then the start rule produced by \repair{} on input $a^m$ is
\begin{center}
$S\to X_{\lfloor\log m\rfloor-1}X_{\lfloor\log m\rfloor-1}f_{\lfloor\log m\rfloor-1}(b_{\lfloor\log m\rfloor-1})\cdots f_1(b_1)f_0(b_0)$.
%\prod_{i=\lfloor\log m\rfloor-1}^0f_i(b_i)$
\end{center}
This means that the symbol $a$ only occurs in the start rule if $b_0=1$, and the nonterminal $X_i$ ($1\le i\le \lfloor\log m\rfloor-2$) occurs in the start rule if and only if
$b_i=1$. Since \repair{} only replaces words with at least two occurrences, the most significant bit $b_{\lfloor \log m\rfloor}=1$ is represented by $X_{\lfloor\log m\rfloor-1}X_{\lfloor\log m\rfloor-1}$.
Note that for $1 \leq m \leq 3$, \repair{} produces the trivial SLP $S \to a^m$.

For the proof of the new lower bound, we use \emph{De Bruijn sequences} \cite{deBr46}. A binary De Bruijn sequence of order $n$ is a string $B_n \in \{0,1\}^*$ of length $2^n$
such that every string from $\{0,1\}^n$ is either a factor of $B_n$ or a suffix of $B_n$ concatenated with a prefix of $B_n$. Moreover, every word of length at least $n$ occurs at most once as factor in $B_n$.
As an example, the string $1100$ is a De Bruijn sequence of order $2$, since $11$,
$10$ and $00$ occur as factors and $01$ occurs as a suffix concatenated with a prefix.
Lemma~\ref{lemma:folklore} (point \ref{mk}) implies that every SLP for $B_n$ has size $\Omega(2^n/n)$.

\begin{theorem} \label{thm}
$\alpha_\mathsf{RePair}(2,n)\in\Omega\left(\log n/\log\log n\right)$
\end{theorem}
\begin{proof}
We start with a binary De Bruijn sequence $B_{\lceil\log k\rceil}\in \{0,1\}^*$ of length $2^{\lceil\log k\rceil}$.
We have $k\le|B_{\lceil\log k\rceil}|< 2k$. Since De Bruijn sequences are not unique, we fix a De Bruijn sequence which starts with $1$ for the remaining proof.
We define a homomorphism $h:\{0,1\}^*\to\{0,1\}^*$ by $h(0)=01$ and $h(1)=10$. The words $w_k$ of length $2k$ are defined as
$$w_k=h(B_{\lceil\log k\rceil}[1:k]).$$
For example $k=4$ and $B_2=1100$ yield $w_4=10100101$.
We will analyze the approximation ratio of \repair{} for the binary words
$$s_k=\prod_{i=1}^{k-1}\left(a^{w_k[1:k+i]}b\right)a^{w_k}=a^{w_k[1:k+1]}ba^{w_k[1:k+2]}b\dots a^{w_k[1:2k-1]}ba^{w_k},$$
where the prefixes $w_k[1:k+i]$ for $1\le i\le k$ are interpreted as integers given by their binary representations. For example we have $s_4=a^{20}ba^{41}ba^{82}ba^{165}$.

Since $B_{\lceil\log k\rceil}[1]=w_k[1]=1$, we have $2^{k+i-1}\le\left|a^{w_k[1:k+i]}\right|\le 2^{k+i}-1$ for $1\le i\le k$ and thus $|s_k|\in \Theta\left(4^k\right)$.
\begin{claim} \label{claim-repair1}
  A smallest SLP producing $s_k$ has size $\bigO(k)$.
\end{claim}
\begin{claimproof}
There is an SLP $\bb A$ of size $\bigO(k)$ for the first $a$-block $a^{w_k[1:k+1]}$ of length $\Theta(2^k)$. % by Theorem~\ref{lemma:folklore}, point~\ref{nlogn}.
Let $A$ be the start nonterminal of $\bb A$.
For the second $a$-block $a^{w_k[1:k+2]}$ we only need one additional rule: If $w_k[k+2]=0$, then we can produce $a^{w_k[1:k+2]}$ by the fresh nonterminal $B$ using the rule $B\to AA$. Otherwise, if $w_k[k+2]=1$, then we use $B\to AAa$. The iteration of that process yields for each $a$-block only one additional rule of size at most $3$. If we replace the $a$-blocks in $s_k$ by nonterminals as described, then the resulting word has size $2k+1$ and hence $g(s_k)\in \bigO(k)$.
\hspace*{\fill}  \mbox{({\em end proof of Claim~\ref{claim-repair1}})}
\end{claimproof}
\begin{claim} \label{claim-repair2}
  The SLP produced by \repair{} on input $s_k$ has size $\Omega(k^2/\log k)$.
\end{claim}
\begin{claimproof}
On unary inputs of length $m$, the start rule produced by \repair{} is strongly related to the binary encoding of $m$ as described above. %Example~\ref{unary}.
On input $s_k$, the algorithm begins to produce a start rule which is similarly related to the binary words $w_k[1:k+i]$ for $1\le i\le k$.
Consider the SLP $\mathbb{G}$ which is produced by \repair{} after $(k-1)$ rounds on input $s_k$. We claim that up to this point \repair{} is not affected by the $b$'s in $s_k$ and therefore has introduced the rules $X_1\to aa$ and $X_i\to X_{i-1}X_{i-1}$ for $2\le i\le k-1$. If this is true, then the first $a$-block is modified in the start rule after $k-1$ rounds as follows
\begin{center}
$S\to X_{k-1}X_{k-1}f_{k-1}(w_k[2])f_{k-2}(w_k[3])\cdots f_0(w_k[k+1])b\cdots$
\end{center}
where $f_0(1)=a$, $f_0(0)=\varepsilon$ and $f_i(1)=X_i$, $f_i(0)=\varepsilon$ for $i\ge 1$.
All other $a$-blocks are longer than the first one, hence each factor of the start rule which corresponds to an $a$-block begins with $X_{k-1}X_{k-1}$. Therefore, the number of occurrences of $X_{k-1}X_{k-1}$ in the SLP is at least $k$. %, which is the number of $a$-blocks.
Since the symbol $b$ occurs only $k-1$ times in $s_k$, it follows that our assumption is correct and \repair{} is not affected by the $b$'s in the first $(k-1)$ rounds on input $s_k$.
Also, for each block $a^{w_k[1:k+i]}$, the $k-1$ least significant bits of $w_k[1:k+i]$ ($1\le i\le k$) are represented in the corresponding factor of the start rule of $\mathbb{G}$,
i.e., the start rule contains non-overlapping factors $v_i$ with
\begin{equation}
v_i=f_{k-2}(w_k[i+2])f_{k-3}(w_k[i+3])\dots f_1(w_k[k+i-1])f_0(w_k[k+i])\label{blockencoding}
\end{equation}
for $1\le i\le k$. For example after $3$ rounds on input $s_4=a^{20}ba^{41}ba^{82}ba^{165}$, we have the start rule
$$S\to \underbrace{X_3X_3X_2}_{a^{20}}b\underbrace{{X_3}^5a}_{a^{41}}b\underbrace{{X_3}^{10}X_1}_{a^{82}}b\underbrace{{X_3}^{20}X_2a}_{a^{165}},$$
where $v_1=X_2$, $v_2=a$, $v_3=X_1$ and $v_4=X_2a$.
The length of the factor $v_i\in\{a,X_1,\dots,X_{k-2}\}^*$ from equation~\eqref{blockencoding} is exactly the number of $1$'s in the word $w_k[i+2:k+i]$. Since $w_k$ is constructed by the homomorphism $h$, it is easy to see that $|v_i|\ge (k-3)/2$.
Note that no letter occurs more than once in $v_i$, hence $g(v_i)=|v_i|$. Further, each substring of length $2\lceil\log k\rceil+2$ occurs at most once in $v_1,\dots,v_k$, because otherwise there would be a factor of length $\lceil\log k\rceil$ occurring more than once in $B_{\lceil\log k\rceil}$. It follows that there are at least
$$k\cdot ( \lceil(k-3)/2\rceil-2\lceil\log k\rceil-1)\in\Theta(k^2)$$
different factors of length $2\lceil\log k\rceil+2\in\Theta(\log k)$ in the right-hand side of the start rule of $\mathbb G$. By Lemma~\ref{lemma:folklore} (point \ref{mk}) it follows that a smallest SLP for the right-hand side of the start rule has size $\Omega(k^2/\log k)$ and therefore $|\mathsf{RePair}(s_k)|\in\Omega(k^2/\log k)$.
\hspace*{\fill}  \mbox{({\em end proof of Claim~\ref{claim-repair2}})}
\end{claimproof}
In conclusion: We showed that a smallest SLP for $s_k$ has size $\bigO(k)$, while
\repair{} produces an SLP of size $\Omega(k^2/\log k)$.
This implies $\alpha_{\mathsf{RePair}}(s_k) \in \Omega(k/\log k)$, which together with $n=|s_k|$ and $k\in\Theta(\log n)$ finishes the proof.
\end{proof}
Note that in the above proof, \repair{} chooses in the first $k-1$ rounds a digram for the replaced maximal string.
Therefore, Theorem~\ref{thm} also holds for the \repair{}-variant, where in every round a digram (which is not necessarily
a maximal string) is replaced.

The goal of this section is to prove the following result:

\begin{theorem} \label{theorem-6}
Let $c \geq 1$ be a constant.
If there exists a polynomial time grammar-based compressor $\mc C$ with
$\alpha_{\mc C}(2,n) \leq c$ then there exists a
polynomial time grammar-based compressor $\mc D$ with
$\alpha_{\mc D}(n) \leq 6c$.
\end{theorem}
For a factor $24+\varepsilon$ (with $\varepsilon > 0$) instead of 6 this result was shown in \cite{ArpeR06} using a more complicated block encoding.

We split the proof of Theorem~\ref{theorem-6} into two lemmas that state translations
between SLPs over arbitrary alphabets
and SLPs over a binary alphabet. For the rest of this section fix the alphabets
$\Sigma = \{ c_0, \ldots, c_{k-1} \}$ and $\Sigma_2 = \{ a, b \}$. To translate
between these two alphabets, we define an injective homomorphism $\varphi \colon \Sigma^* \to \Sigma_2^*$ by
\begin{equation} \label{home-varphi}
\varphi(c_i) = a^ib \quad (0 \leq i \leq k-1).
\end{equation}

\begin{lemma} \label{lemma-factor-3}
Let $w \in \Sigma^*$ be such that every symbol from $\Sigma$
occurs in $w$. From an SLP $\bb A$ for $w$ one can construct in polynomial time an
SLP $\bb B$ for $\varphi(w)$ of size at most $3 \cdot |\bb A|$.
\end{lemma}

\begin{proof}
To translate $\bb A$ into an SLP $\bb B$ for $\varphi(w)$, we first
add the productions $A_0 \to b$ and $A_i \to aA_{i-1}$ for every $i,1 \leq i \leq k-1$.
%These productions have a total size of $2k-1$.
Finally, we replace in $\bb A$ every
occurrence of $c_i \in \Sigma$ by $A_i$. This yields an SLP $\bb B$ for
$\varphi(w)$ of size $|\bb A| + 2k-1$.  Because $k \le |\bb A|$
(since every symbol from $\Sigma$ occurs in $w$), we
obtain $|\bb B| \le 3 \cdot |\bb A|$.
\end{proof}

\begin{lemma} \label{lemma-factor-2}
Let $w \in \Sigma^*$ such that every symbol from $\Sigma$
occurs in $w$. From an SLP $\bb B$ for $\varphi(w)$ one can construct in polynomial time an
SLP $\bb A$ for $w$ of size at most $2 \cdot |\bb B|$.
\end{lemma}

\begin{proof}
Let $\bb B=(N,\Sigma_2,P,S)$ be an SLP for $\varphi(w)$, where $w \in \Sigma^*$. 
We can assume that every right-hand side of $\bb B$ is a non-empty string.
Consider a nonterminal $A \in N$ of $\bb B$.
Since $\bb B$ produces $\varphi(w)$, $A$ produces a factor of $\varphi(w)$, which is a word from $\{a,b\}^*$.
We cannot directly translate $\val(A)$ back to a word over $\Sigma^*$ because $\val(A)$ does not have to belong
to the image of $\varphi$. But $\val(A)$ is a factor of a string from $\varphi(\Sigma^*)$.
%Consider for instance the SLP with the productions $S \to AB$, $A \to aa$, and $B \to abab$. Then $S$ has to be translated by $c_3c_1$,
%where the $aa$ from $\val(A)$ and the leftmost $a$ from $\val(B)$ are merged into $c_3$.
Note that a string over $\{a,b\}$ is a factor of a string from $\varphi(\Sigma^*)$ if and only if it does not contain
a factor $a^i$ with $i \geq k$.
Let $\val(A) = a^{i_1}b   \cdots a^{i_n} b a^{i_{n+1}}$ be such a string, where $n \geq 0$, and $0 \leq i_1, \ldots, i_{n+1} < k$.
We factorize $\val(A)$ into three parts in the following way.
If $n=0$ (i.e., $\val(A) = a^{i_1}$) then we split $\val(A)$ into $\varepsilon$, $\varepsilon$, and $a^{i_1}$.
If $n >0$ then we split  $\val(A)$ into $a^{i_1}b$, $a^{i_2}b   \cdots a^{i_n} b$, and $a^{i_{n+1}}$.
Let us explain the intuition behind this factorization. We concentrate on the case
$n>0$; the case $n=0$ is simpler. Note that irrespective of the context in which an occurrence of $\val(A)$ appears
in $\val(\bb B)$, we can translate the middle part $a^{i_2}b   \cdots a^{i_n} b$ into $c_{i_2} \cdots c_{i_n}$.
We will therefore introduce in the SLP $\bb A$ for $w$ a variable $A'$ that produces $c_{i_2} \cdots c_{i_n}$.
For the left part $a^{i_1}b$ we can not directly produce $c_{i_1}$ because an occurrence of $\val(A)$ could be preceded by
an $a$-block $a^{i_0}$, yielding the symbol $c_{i_0+i_1}$. Therefore, the algorithm that produces $\bb A$ will only
memorize the symbol $c_{i_1}$ without writing it directly on a right-hand side of an $\bb A$-production. Similarly, the algorithm
will memorize the length $i_{n+1}$ of the final $a$-block of $\val(A)$.

Let us now come to the formal details of the proof.
As usual, we write $\mathbb{Z}_k$ for $\{0,1,\ldots,k-1\}$ and w.l.o.g. we assume that $\Sigma \cap \mathbb{Z}_k = \emptyset$.
Consider a word  $s = a^{i_1}b   \cdots a^{i_n} b a^{i_{n+1}}$, where $n \geq 0$, and $0 \leq i_1, \ldots, i_{n+1} < k$.
Motivated by the above discussion, we define $\ell(s)\in \Sigma \cup \{ \varepsilon \}$, $m(s) \in \Sigma^*$ and $r(s)\in \mathbb{Z}_k$
 as follows:
\begin{align*}
\ell(s) & =  \begin{cases}
     c_{i_1} & \text{ if } n \geq 1, \\
     \varepsilon & \text{ if } n=0,
     \end{cases}
\\
m(s) &= c_{i_2} \cdots c_{i_n}, \\
r(s) &= i_{n+1} .
\end{align*}
 Note that $\ell(s)=\varepsilon$ implies $m(s)= \varepsilon$.
 Finally, we define the word $\psi(s) \in \Sigma^* \mathbb{Z}_k$ as
\[
  \psi(s) = \ell(s) m(s) r(s) .
\]
For a nonterminal $A \in N$ we define $\ell(A) = \ell(\val(A))$, $m(A) = m(\val(A))$ and
$r(A) =  r(\val(A))$. We now define an SLP $\bb A'$ that contains
for every nonterminal $A \in N$ a nonterminal $A'$ such that $\val(A') = m(A)$.
Moreover, the algorithm also computes  $\ell(A) \in \Sigma \cup \{ \varepsilon \}$ and $r(A) \in \mathbb{Z}_k$.

We define the productions of $\bb A'$ inductively over the structure of $\bb B$.
Consider a production $(A \to \alpha) \in P$, where
$\alpha =  v_0 A_1 v_1 A_2 \cdots v_{n-1} A_n v_n \neq \varepsilon$
with $n \geq 0$, $A_1, \ldots, A_n \in N$, and $v_0, v_1, \ldots, v_n \in \Sigma_2^*$.
Let $\ell_i = \ell(A_i) \in \Sigma \cup \{\varepsilon\}$ and $r_i = r(A_i) \in \mathbb{Z}_k$, which have already been computed.
The right-hand side for $A'$ is obtained as follows.
We start with the word
\begin{equation} \label{rhs-factor-2}
\psi(v_0) \,  \ell_1 \, A'_1 \, r_1 \, \psi(v_1) \, \ell_2 \, A'_2 \, r_2 \cdots \psi(v_{n-1}) \, \ell_n \, A'_n \, r_n \, \psi(v_n) .
\end{equation}
Note that each of the factors $\ell_i A'_i r_i$ produces (by induction) $\psi(\val(A_i))$.
Next we remove every $A'_i$ that derives the empty word (which is equivalent to $m(A_i) = \varepsilon$).
After this step, every occurrence of a symbol $\rsym{i} \in \mathbb{Z}_k$ in \eqref{rhs-factor-2} is either the last symbol of the above word or it is
followed by a symbol from $\mathbb{Z}_k \cup \Sigma$ (but not followed by a nonterminal $A'_j$). To see this, recall
that $\ell_j = \varepsilon$ implies $m(A_j) = \varepsilon$, in which case $A'_j$ is removed in \eqref{rhs-factor-2}.

The above fact allows us
to eliminate all occurrences of symbols $\rsym{i} \in \mathbb{Z}_k$ in \eqref{rhs-factor-2} except for the last one using the two reduction
rules $\rsym{i}\, \rsym{j} \to \rsym{i+j}$ for $i,j \in \mathbb{Z}_k$ (which corresponds to $a^i a^j = a^{i+j}$)
and $\rsym{i}\, c_j \to c_{i+j}$ (which corresponds to $a^i a^jb = a^{i+j}b$).
If we perform these rules as long as possible (the order of applications is not relevant
since these rules form a confluent and terminating system), only a single occurrence
of a symbol $\rsym{i} \in \mathbb{Z}_k$ at the end of the string will remain.
The resulting string $\alpha'$ produces $\psi(A)$.
Hence, we obtain the right-hand side for the nonterminal $A'$ by removing the first symbol of $\alpha'$ if it is
from $\Sigma$ (this symbol is then $\ell(A)$) and the last symbol of $\alpha'$,
which must be from $\mathbb{Z}_k$ (this symbol is $r(A)$).
Note that if $\alpha'$ does not start with a symbol from $\Sigma$, then
$\alpha'$ belongs to $\mathbb{Z}_k$, in which case we have $\ell(A) = \varepsilon$.

Note that $\psi(\varphi(w)) = w\rsym{0}$ for every $w \in \Sigma^*$,
so for the start variable $S$ of $\bb B$ we must have $r(S) = \rsym{0}$,
since $\val_{\bb B}(S) \in \varphi(\Sigma^*)$.
Let $S' \to \sigma$ be the production for $S'$ in $\bb A'$.
We obtain the SLP $\bb A$ by replacing this production by $S' \to \ell(S) \sigma$.
Since $\val_{\bb A'}(S')=m(S)$ and $\val_{\bb B}(S) = \varphi(w)$ we have $\val_{\bb A}(S') = \ell(S) m(S) = w$.

To bound the size of $\bb A$ consider the word in \eqref{rhs-factor-2} from which the right-hand side of
the nonterminal $A'$ is computed. All occurrences of symbols from $\mathbb{Z}_k$ are eliminated when
forming this right-hand side. This leaves a word of length at most $|\alpha|+n$ (where $\alpha$ is the original
right-hand side of the nonterminal $A$). The additive term $n$ comes from the symbols $\ell_1, \ldots, \ell_n$.
Hence, $|\bb A'|$ is bounded by the size of $\bb B$ plus the total number of
occurrences of nonterminals in right-hand sides of $\bb B$, which is at most $2 |\bb B| -1$ (there is at least one terminal
occurrence in a right-hand side). Since $|\bb A| = |\bb A'|+1$ we get $|\bb A| \leq 2 |\bb B|$.

The algorithm's runtime for a production $A \to \alpha$ is linear in $|\alpha|$.
This is because we start with the string \eqref{rhs-factor-2} which can be computed in time $\bigO(|\alpha|)$.
From this string, we remove all the $A_i'$ that produce $\varepsilon$ and we also apply the two rewriting rules.
Both of these can be done in a single left-to-right sweep over the string.
The number of operations needed is linear in $|\alpha|$,
where each operation needs constant time,
i.e. removing an $A_i'$ takes constant time, and using one of the rewriting rules also takes constant time.
Since the algorithm uses the structure of $\bb B$ to visit each of its productions once,
we overall obtain a linear running time in the size of $\bb B$.
\end{proof}

\begin{example}
Consider the production  $A \to a^3 b a^5 A_1 a^3 A_2 a^2 b^2 A_3 a^2$
and assume that $\val(A_1) = a^2$, $\val(A_2) = a b a^3 b a$ and $\val(A_3) = b a^2b a^3$.
Hence, when we produce the right-hand side for $A'$ we have:
$\val(A'_1) = \varepsilon$, $\val(A'_2) = c_3$, $\val(A'_3) = c_2$, $\ell_1 = \varepsilon$, $r_1 = \rsym{2}$,
$\ell_2 = c_1$, $r_2 = \rsym{1}$, $\ell_3 = c_0$, $r_3 = \rsym{3}$.
We start with the word (every digit is a single symbol)
\[
c_3 \, 5\, A'_1 \, 2 \, 3 \, c_1 A'_2 \,1 \, c_2 c_0 \, 0 \,  c_0 A'_3 \, 3\, 2 .
\]
Then we replace $A'_1$ by $\varepsilon$ and obtain
$c_3 \, 5\, 2 \, 3 \, c_1 A'_2 \,1 \, c_2 c_0 \, 0 \,  c_0 A'_3 \, 3\, 2$.
Applying the reduction rules finally yields
$c_3 c_{11} A'_2 c_3 c_0 c_0 A'_3 \rsym{5}$. Hence, we have $\ell(A) = c_3$, $r(A) = \rsym{5}$ and the production for $A'$ is
$A' \to c_{11} A'_2 c_3 c_0 c_0 A'_3$.
\end{example}
{\it Proof of Theorem~\ref{theorem-6}.} Let $\mc C$ be an arbitrary grammar-based compressor working
in polynomial time such that $\alpha_{\mc C}(2,n) \leq c$. The grammar-based compressor $\mc D$ works for an input word $w$ over an arbitrary
alphabet as follows: Let $\Sigma = \{c_0, \ldots, c_{k-1}\}$ be the set of symbols that occur in $w$ and let $\varphi$
be defined as in \eqref{home-varphi}. Using $\mc C$, one first computes an SLP $\bb B$ for $\varphi(w)$ such that $|\bb B| \leq c \cdot g(\varphi(w))$.
Then, using Lemma~\ref{lemma-factor-2}, one computes from $\bb B$ an SLP $\bb A$ for $w$ such that
$|\bb A| \leq 2c \cdot g(\varphi(w))$.
Lemma~\ref{lemma-factor-3} implies $g(\varphi(w)) \leq 3 \cdot g(w)$ and hence
$|\bb A| \leq 6c \cdot g(w)$, which proves the theorem.
\qed

\section{Hardness of grammar-based compression for binary alphabets}

\section{Open problems}

One should try to narrow the gaps between the lower and
 upper bounds for the other grammar-based compressors analyzed in \cite{CLLLPPSS05}.
 In particular, the gap between the known lower and upper bounds for the so-called global algorithms from \cite{CLLLPPSS05} (like \repair{}) is still quite big.
Charikar et al.~\cite{CLLLPPSS05} prove an upper bound $\Theta( (n/\log n)^{2/3})$ for every global algorithm and nothing better is known for the three
global algorithms \repair{}, {\sf LongestMatch}, {\sf Greedy} studied in \cite{CLLLPPSS05}. Comparing to this upper bound, the known lower bounds are quite
small: $\Omega(\log n / \log\log n)$ for \repair{} (by our Theorem~\ref{thm}), $\Omega(\log\log n)$ for longest match \cite{CLLLPPSS05}, and $1.348\ldots$. The latter is
a very recent result from \cite{Hu19}.\footnote{The table on page 2556 in \cite{CLLLPPSS05} states the better lower bound of $1.37\ldots$, but the authors only show the lower
bound $1.137\ldots$, see \cite[Theorem 11]{CLLLPPSS05}.}

Another open research problem is improving the constant 6 in Theorem~\ref{theorem-6}. Recall that lowering this constant to at most
$8569/8568$ would imply that the smallest grammar problem for binary strings cannot be solved in polynomial time unless $\mathsf{P} = \mathsf{NP}$.

%\bibliographystyle{plainurl}
%\bibliography{bib}

\begin{thebibliography}{10}

\bibitem{ApostolicoL98}
Alberto Apostolico and Stefano Lonardi.
\newblock Some theory and practice of greedy off-line textual substitution.
\newblock In {\em Proceedings of {DCC} 1998}, pages 119--128. {IEEE} Computer
  Society, 1998.

\bibitem{ArpeR06}
Jan Arpe and R{\"u}diger Reischuk.
\newblock On the complexity of optimal grammar-based compression.
\newblock In {\em Proceedings of Data Compression Conference (DCC 2006)}, pages
  173--182. IEEE Computer Society, 2006.

\bibitem{BerstelB87}
Jean Berstel and Srecko Brlek.
\newblock On the length of word chains.
\newblock {\em Inf. Process. Lett.}, 26(1):23--28, 1987.

\bibitem{BilleGP17}
Philip Bille, Inge~Li G{\o}rtz, and Nicola Prezza.
\newblock Space-efficient {Re-Pair} compression.
\newblock In {\em Proceedings of {DCC} 2017}, pages 171--180, 2017.

\bibitem{BilleLRSSW15}
Philip Bille, Gad~M. Landau, Rajeev Raman, Kunihiko Sadakane, Srinivasa~Rao
  Satti, and Oren Weimann.
\newblock Random access to grammar-compressed strings and trees.
\newblock {\em {SIAM} Journal on Computing}, 44(3):513--539, 2015.

\bibitem{CaFeGaGrSchmi16}
Katrin Casel, Henning Fernau, Serge Gaspers, Benjamin Gras, and Markus~L.
  Schmid.
\newblock On the complexity of grammar-based compression over fixed alphabets.
\newblock In {\em Proceedings of ICALP 2016}, Lecture Notes in Computer
  Science. Springer, 1996.
\newblock to appear.

\bibitem{CLLLPPSS05}
M.~Charikar, E.~Lehman, A.~Lehman, D.~Liu, R.~Panigrahy, M.~Prabhakaran,
  A.~Sahai, and A.~Shelat.
\newblock The smallest grammar problem.
\newblock {\em IEEE Trans.~Inf.~Theory}, 51(7):2554--2576, 2005.

\bibitem{ClaudeN10}
Francisco Claude and Gonzalo Navarro.
\newblock Fast and compact web graph representations.
\newblock {\em {ACM Transactions on the Web}}, 4(4):16:1--16:31, 2010.

\bibitem{deBr46}
Nicolaas de~Bruijn.
\newblock A combinatorial problem.
\newblock {\em Proc. Koninklijke Nederlandse Akademie v. Wetenschappen}, pages
  758--764, 1946.

\bibitem{Diw86}
A.A. Diwan.
\newblock A new combinatorial complexity measure for languages.
\newblock Tata Institute, Bombay, India, 1986.

\bibitem{FuruyaTNIBK19}
Isamu Furuya, Takuya Takagi, Yuto Nakashima, Shunsuke Inenaga, Hideo Bannai,
  and Takuya Kida.
\newblock Mr-repair: Grammar compression based on maximal repeats.
\newblock In {\em Proceedings of {DCC} 2019}, pages 508--517. {IEEE}, 2019.

\bibitem{GaIMaNaSaTa}
Travis Gagie, Tomohiro I, Giovanni Manzini, Gonzalo Navarro, Hiroshi Sakamoto,
  and Yoshimasa Takabatake.
\newblock Rpair: Rescaling {RePair} with {Rsync}.
\newblock {\em CoRR}, abs/1906.00809, 2019.
\newblock URL: \url{http://arxiv.org/abs/1906.00809}.

\bibitem{Ganczorz18}
Micha\l{} Ga\'nczorz.
\newblock Entropy bounds for grammar compression.
\newblock {\em CoRR}, abs/1804.08547, 2018.
\newblock URL: \url{http://arxiv.org/abs/1804.08547}.

\bibitem{Ganczorz19}
Micha\l{} Ga\'nczorz.
\newblock Entropy lower bounds for dictionary compression.
\newblock In {\em Proceedings of {CPM} 2019}, volume 128 of {\em LIPIcs}, pages
  11:1--11:18. Schloss Dagstuhl - Leibniz-Zentrum f\"ur Informatik, 2019.

\bibitem{GanczorzJ17}
Micha\l{} Ga\'nczorz and Artur Je\.{z}.
\newblock Improvements on {R}e-{P}air grammar compressor.
\newblock In {\em Proceedings of {DCC} 2017}, pages 181--190. {IEEE}, 2017.

\bibitem{HuLoRe17}
Danny Hucke, Markus Lohrey, and Carl~Philipp Reh.
\newblock The smallest grammar problem revisited.
\newblock In Proceedings of {\em {SPIRE}} 2017, volume 9954 of {\em {LNCS}}, pages 35--49, 2016.
\newblock URL: \url{https://doi.org/10.1007/978-3-319-46049-9\_4}.

\bibitem{Hu19}
Danny Hucke.
\newblock Approximation ratios of RePair, LongestMatch and Greedy on unary strings.
\newblock to appear in Proceedings of {\em {SPIRE}} 2019.


\bibitem{Jez13approx}
Artur Je\.z.
\newblock Approximation of grammar-based compression via recompression.
\newblock {\em Theoretical Computer Science}, 592:115--134, 2015.

\bibitem{Jez16}
Artur Je\.z.
\newblock A really simple approximation of smallest grammar.
\newblock {\em Theoretical Computer Science}, 616:141--150, 2016.

\bibitem{KempaP18}
Dominik Kempa and Nicola Prezza.
\newblock At the roots of dictionary compression: string attractors.
\newblock In {\em Proceedings of {STOC} 2018}, pages 827--840. {ACM}, 2018.

\bibitem{KiYa00}
J.~C. {Kieffer} and E.-H. {Yang}.
\newblock {Grammar-based codes: A new class of universal lossless source
  codes.}
\newblock {\em IEEE Trans.~Inf.~Theory}, 46(3):737--754, 2000.

\bibitem{KiefferYNC00}
J.~C. Kieffer, E.-H. Yang, G.~J. Nelson, and P.~C. Cosman.
\newblock Universal lossless compression via multilevel pattern matching.
\newblock {\em IEEE Trans.~Inf.~Theory}, 46(4):1227--1245, 2000.

\bibitem{KiFlaYa11}
John~C. Kieffer, Philippe Flajolet, and En{-}Hui Yang.
\newblock Universal lossless data compression via binary decision diagrams.
\newblock {\em CoRR}, abs/1111.1432, 2011.
\newblock URL: \url{http://arxiv.org/abs/1111.1432}.

\bibitem{KiYa02}
John~C. {Kieffer} and En~hui {Yang}.
\newblock Structured grammar-based codes for universal lossless data
  compression.
\newblock {\em Communications in Information and Systems}, 2(1):29--52, 2002.

\bibitem{KosarajuM99}
S.~Rao Kosaraju and Giovanni Manzini.
\newblock Compression of low entropy strings with {Lempel-Ziv} algorithms.
\newblock {\em {SIAM} Journal on Computing}, 29(3):893--911, 1999.

\bibitem{DBLP:conf/dcc/LarssonM99}
N.~J. Larsson and A.~Moffat.
\newblock Offline dictionary-based compression.
\newblock In {\em Proc.~DCC 1999}, pages 296--305. IEEE, 1999.

\bibitem{LohreyMM13}
M.~Lohrey, S.~Maneth, and R.~Mennicke.
\newblock {XML} tree structure compression using {RePair}.
\newblock {\em Inform.~Syst.}, 38(8):1150--1167, 2013.

\bibitem{Loh12survey}
Markus Lohrey.
\newblock Algorithmics on {SLP}-compressed strings: A survey.
\newblock {\em Groups Complexity Cryptology}, 4(2):241--299, 2012.

\bibitem{MasakiK16}
Takuya Masaki and Takuya Kida.
\newblock Online grammar transformation based on re-pair algorithm.
\newblock In {\em Proceedings of {DCC} 2016}, pages 349--358. {IEEE}, 2016.

\bibitem{Nevill-ManningW97}
Craig~G. Nevill-Manning and Ian~H. Witten.
\newblock Identifying hierarchical strcture in sequences: A linear-time
  algorithm.
\newblock {\em J. Artif. Intell. Res. (JAIR)}, 7:67--82, 1997.

\bibitem{OchoaN19}
Carlos Ochoa and Gonzalo Navarro.
\newblock Repair and all irreducible grammars are upper bounded by high-order
  empirical entropy.
\newblock {\em {IEEE} Transactions on Information Theory}, 65(5):3160--3164,
  2019.

\bibitem{rubin76}
Frank Rubin.
\newblock Experiments in text file compression.
\newblock {\em Commun. {ACM}}, 19(11):617--623, 1976.
\newblock URL: \url{http://doi.acm.org/10.1145/360363.360368}.

\bibitem{Ryt03}
W.~Rytter.
\newblock Application of {Lempel}-{Ziv} factorization to the approximation of
  grammar-based compression.
\newblock {\em Theor.~Comput.~Sci.}, 302(1--3):211--222, 2003.

\bibitem{Sakamoto05}
Hiroshi Sakamoto.
\newblock A fully linear-time approximation algorithm for grammar-based
  compression.
\newblock {\em J. Discrete Algorithms}, 3(2-4):416--430, 2005.

\bibitem{StorerS82}
James~A. Storer and Thomas~G. Szymanski.
\newblock Data compression via textual substitution.
\newblock {\em J. {ACM}}, 29(4):928--951, 1982.

\bibitem{TabeiTS13}
Yasuo Tabei, Yoshimasa Takabatake, and Hiroshi Sakamoto.
\newblock A succinct grammar compression.
\newblock In {\em Proceedings of the 24th Annual Symposium on Combinatorial
  Pattern Matching, CPM 2013}, volume 7922 of {\em Lecture Notes in Computer
  Science}, pages 235--246. Springer, 2013.

\bibitem{YangK00}
En-Hui Yang and John~C. Kieffer.
\newblock Efficient universal lossless data compression algorithms based on a
  greedy sequential grammar transform - part one: Without context models.
\newblock {\em IEEE Transactions on Information Theory}, 46(3):755--777, 2000.

\bibitem{Yao76}
Andrew Chi-Chih Yao.
\newblock On the evaluation of powers.
\newblock {\em SIAM Journal on Computing}, 5(1):100--103, 1976.

\bibitem{ZiLe78}
Jacob Ziv and Abraham Lempel.
\newblock Compression of individual sequences via variable-rate coding.
\newblock {\em IEEE Transactions on Information Theory}, 24(5):530--536, 1977.

\bibitem{Wan03}
R. Wan.
\newblock Browsing and Searching Compressed Documents.
\newblock PhD thesis, Dept. of Computer Science and Software Engineering, University of Melbourne, 2003.

\bibitem{Kida03}
T. Kida, T. Matsumoto, Y. Shibata, M. Takeda, A. Shinohara, and S. Arikawa.
\newblock Collage systems: a unifying framework for compressed pattern matching.
\newblock {\em Theoretical Computer Science}, 298(1):253–272, 2003.

\bibitem{Gonzalez07}
R. Gonz\'{a}lez and G. Navarro.
\newblock Compressed text indexes with fast locate.
\newblock In {\em Proceedings of the 18th Annual Symposium on Combinatorial Pattern Matching, CPM 2007},
volume 4580 of {\em Lecture Notes in Computer Science}, pages 216–-227. Springer, 2007.

\end{thebibliography}

\end{document}